\documentclass[twocolumn,preprintnumbers,amsmath,amssymb,showpacs, superscriptaddress]{revtex4}

\usepackage{graphicx}% Include figure files
\usepackage{dcolumn}% Align table columns on decimal point
\usepackage{bm}% bold math
\usepackage{bbm}
\usepackage{amsmath}
\usepackage{amsthm}
\usepackage{dsfont}
\usepackage{amssymb}

\newcommand{\id}{\mathbbm{1}}

\newtheorem{thm}{Theorem}

\newtheorem{lem}[thm]{Lemma}
\newtheorem{defin}{Definition}

\let\oldmarginpar\marginpar
\renewcommand\marginpar[1]{\-\oldmarginpar[\raggedleft\marginparsize #1]%
{\raggedright\marginparsize #1}}

\begin{document}

\setlength{\tabcolsep}{1ex}%for latex tables

\title{Tsirelson's bound from a Generalised Data Processing Inequality}% Force line breaks with \\
\author{Oscar C. O. Dahlsten}
\affiliation{Institute for Theoretical Physics, ETH Z\"urich, 8093 Zurich, Switzerland}
\affiliation{Center for Quantum Technologies, National University of Singapore, Republic of Singapore}
\affiliation{Atomic and Laser Physics, Clarendon Laboratory,
University of Oxford, Parks Road, Oxford OX13PU, United Kingdom}

\author{Daniel Lercher}
\affiliation{Institute for Theoretical Physics, ETH Z\"urich, 8093 Zurich, Switzerland}
\affiliation{Department of Mathematics, Technische Universit\"at M\"unchen, 85748
Garching, Germany}
\author{Renato Renner }
\affiliation{Institute for Theoretical Physics, ETH Z\"urich, 8093 Zurich, Switzerland}

% Include the date command, but leave its argument blank.
\date{\today}

\begin{abstract}
The strength of quantum correlations is bounded from above by Tsirelson's bound. We establish a connection between this bound and the fact that correlations between two systems cannot increase under local operations, a property known as the \emph{data processing inequality}. More specifically, we consider arbitrary convex probabilistic theories. These can be equipped with an entropy measure that naturally generalizes the von Neumann entropy, as shown recently in [Short and Wehner, Barnum et. al.]. We prove that if the data processing inequality holds with respect to this generalized entropy measure then the underlying theory necessarily respects Tsirelson's bound. We moreover generalise this statement to any entropy measure satisfying certain minimal requirements. A consequence of our result is that not all of the entropic relations used to derive Tsirelson's bound via information causality in [Pawlowski et. al.] are necessary.   
\end{abstract}

%%%%%%%%%%%%%%%%% END OF PREAMBLE %%%%%%%%%%%%%%%%


\maketitle 

{\em \bf Introduction.---}Quantum mechanics departs fundamentally from any classical theory by allowing non-local correlations~\cite{Bell64}. The  existence of such correlations has been extensively verified in experiments (up to a few loopholes), see e.g.\ \cite{Aspect99}. As was shown by Bell, these correlations imply that the world is not both local and realist, two standard assumptions underpinning the classical mechanical world-view~\cite{Bell64}. Apart from their fundamental theoretical interest, non-local correlations are also of technological importance, for example as the essential ingredient in Ekert-style quantum cryptographic schemes~\cite{Ekert91}. 

However there is a limit to how much local realism is violated. The strength of quantum correlations are themselves upper bounded by Tsirelson's bound~\cite{Tsirelson93,Cirelson80}. This is a non-trivial bound, because it is conceivable to violate Bell-inequalities more than quantum theory, without having a theory which is signalling (allows instantaneous information transfer across space). For example it is possible to conceive of {\em PR-boxes}, also known as {\em non-local boxes}, hypothetical systems which maximally violate the CHSH Bell-inquality, without being signalling~\cite{PopescuR94}. 

The question then arises as to whether one can associate a fundamental assumption about nature other than non-signalling with Tsirelson's bound. Such an assumption could then be labelled a fundamental principle underpinning quantum theory, and possibly form part of a much-sought-for set of principles from which quantum theory could be derived.

There has already been significant effort in this direction. For example, it is now known that the existence of maximally Bell-violating correlations would lead to some communication complexity problems becoming trivial~\cite{Buhrman06, BrassardBLMTU06}, the possibility of oblivious transfer~\cite{vanDam05}, weaker uncertainty relations~\cite{OppenheimW10}, general invalidation of quantum theory locally~\cite{BarnumBBEW10} and severely limited dynamics~\cite{Barrett07,GrossMCD10}. 

A recent string of related papers have moreover been concerned with a principle called information causality~\cite{PawlowskiPKSWZ09,AllcockBPS09, CavalcantiSS10}. A great advantage with this principle is that the exact Tsirelson's bound is recovered, i.e.\ it rules out any stronger correlations, not just the maximally strong ones. The principle amounts to placing a limit on how well two separated parties can perform in a particular game (van Dam's game~\cite{vanDam05}) where they share a resource state. This limits the resource state in such a way that Tsirelson's bound is recovered. Whilst the original interpretation of information causality as a particularly simple generalisation of non-signalling has been questioned (see e.g.~\cite{ShortW09}), the principle is ---as mentioned above--- powerful. 

\begin{figure}%[htb]
 \centering \includegraphics[width=\linewidth]{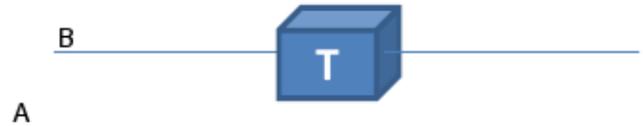}
 \caption{The Data Processing Inequality states that the correlations between $A$ and $B$ cannot increase under a local operation $T$ on $B$. More specifically $H(A|B)\leq H(A|T(B))$.} \label{fig:DPI} \end{figure}

Intriguingly, in the proof that information causality holds in quantum theory, a specific limited set of information-theoretic theorems are used. One may thus replace information causality as a postulate with those information-theoretic theorems. This is attractive if one seeks an information-theoretic set of principles for quantum theory. In order to discuss the validity of such theorems outside of quantum theory, however, one needs definitions of the relevant entropies for general probabilistic theories. Fortuitously, such definitions were recently proposed and investigated in~\cite{ShortW09,Barnum09,KimuraNI09}. In \cite{ShortW09,Barnum09} information causality is also discussed. In~\cite{Barnum09} three sufficient conditions under which a generalised probabilistic theory respects information causality are determined. In~\cite{ShortW09} it is shown that if one follows the information causality proof in the case of {\em box-world}, the theory with PR-boxes and all other non-signalling distributions, the proof breaks down at the point where one needs to assume the so-called strong subadditivity of entropy. An alternative approach to deriving information causality from more basic entropic principles appears in~\cite{HsuIF11}. These recent works when taken together suggest that one may hope for a small and operationally motivated set of information-theoretic relations from which Tsirelson's bound, and perhaps even quantum theory, can be derived.

We here investigate the Data Processing Inequality as such a principle. This essentially states that correlations, quantified via conditional entropies,  cannot increase under local operations, see Fig.~\ref{fig:DPI}. In order to define this in general, we use an entropy proposed in~\cite{ShortW09}, which naturally generalizes the von Neumann entropy (and reduces to the latter in the case of quantum theory). We prove that, surprisingly, this generalised Data Processing Inequality alone implies Tsirelson's bound.

We proceed as follows. We firstly describe the framework of generalised probabilistic theories within which we work. Then we define Tsirelson's bound as well as information causality. We go on to describe how to define entropy in an operational manner as in~\cite{ShortW09}. This is used to define the generalised Data Processing Inequality (DPI). We then prove that DPI implies Tsirelson's bound. This involves proving a more general theorem of which the main result is a corollary. Finally we compare the results to previous ones and discuss the implications and interpretation of the principle.

{\em \bf Convex, operational, probabilistic theories.---}We use the framework of convex probabilistic theories~\cite{Barrett07,BarnumBLW06,Hardy01}. This amounts to taking the minimalistic pragmatic view that the operational content of a theory is in the predicted statistics of measurement outcomes. 

The {\em state} of a system by definition determines the probabilities of all possible measurement outcomes. The state is  completely specified, again by definition, by the probabilities for the outcomes of $k$ so-called {\em fiducial} measurements $0,\dots,{k-1}$. $k$ may be significantly smaller than the total number of measurements (e.g.\ in quantum theory there is a continuum of measurements but $k=d^2$ for a state on a Hilbert space of dimension d). If these fiducial measurements 
each have $l$ possible outcomes $0,\dots,l-1$ we will say that the system is of type $(k,l)$.

We can thus write a (normalised) state as a list of $P(i|j)$, denoting the probability of getting outcome $i$ if fiducial measurement $j$ is performed. We represent this by $\vec{P}$. The normalisation of the state is $|\vec{P}|:=\sum_i P(i|j)$ and is for all valid states independent of the choice of fiducial measurement $j$. A state is said to be normalised if $|\vec{P}|=1$ and subnormalised if $|\vec{P}|< 1$.

We assume that the set of allowed normalized states $\mathcal{S}$ is closed and convex (so that any probabilistic mixture of states is an allowed state). We say a state is {\em pure} if it cannot be written as a convex mixture of other states. 
A {\em theory} is defined by the set of allowed states, $\mathcal{S}$, as well as the set of allowed transformations.

Transformations take states to states. They must be linear as probabilistic mixtures of different states must be conserved~\cite{Barrett07}. Transformations can thus be modelled as $\vec{P} \mapsto M\cdot\vec{P}$, where $M$ is a matrix. If one performs a measurement with several outcomes each outcome is associated with a certain transform $M_i$. The unnormalized state associated with the $i$-th outcome is  $M_i\cdot\vec{P}$, and the associated probability of the $i$-th outcome is given by the normalisation factor after the transformation: $|M_i\cdot\vec{P}|$.

If one is only interested in the probabilities of the different outcomes of a measurement one can always associate
 with a transformation $\{M_i\}$ a set of vectors $\{R_i\}$ such that $\vec{R}_i\cdot\vec{P}=|M_i\cdot \vec{P}|\, \forall \, \vec{P} \in \mathcal{S}$. Consequently, for a normalized state $\vec{P}$, $\vec{R}_i\cdot\vec{P}$ is the probability of the $i$th outcome.

It is also possible to combine single systems to form multipartite systems. If one performs local operations on the systems $A$ and $B$ the final unnormalized state of the joint system does by assumption not depend on the temporal ordering of the operations. A direct consequence of this is the no-signaling principle: measuring system $B$ cannot give information about what transformation was applied on $A$~\cite{Barrett07}.

We will make the non-trivial but standard assumption that the global state of a bipartite system can be completely determined by specifying joint probabilities of outcomes for fiducial measurements performed simultaneously on each subsystem. Accordingly, the joint state of two parties is uniquely specified by the list $P(ii'|jj')$, denoting the probability of getting the outcomes $i$ and $i'$ if one performs fiducial measurement $j$ on $A$
and $j'$ on $B$.

For a joint state $\vec P_{AB}$, the {\em marginal} (also called {\em reduced}) state of system $A$, denoted $\vec P_A$, is given by 
$P_A(i|j)\equiv \sum_{i'} P_{AB}(ii'|jj')$. Similarly, the {\em conditional marginal} state $\vec P_{A|B:k,l}$ is defined by
\begin{equation}
 \label{def_condmargstate}
P_{A|B:k,l}(i|j)\equiv \frac{P_{AB}(ik|jl)}{P_B(k|l)}\,.
\end{equation}
This represents the state of system $A$ after a fiducial measurement $l$ was performed on system $B$ and the outcome $k$ was obtained.  

It was shown in~\cite{Barrett07} that denoting the vector spaces containing the vectors $\vec P_{AB}$, $\vec P_A$, and $\vec P_B$ by 
$V_{AB}$, $V_A$ and $V_B$, respectively, one can relate the spaces by $V_{AB}=V_A\otimes V_B$ ($\otimes$ being the tensor product). 
One assumes that for $\vec{P}_A \in \mathcal{S}_A$ and $\vec{P}_B \in \mathcal{S}_B$ we have 
$\vec{P}_A \otimes \vec{P}_B \in \mathcal{S}_{AB}$. This implies that any $\vec P_{AB}\in\mathcal{S}_{AB}$ can be written as 
$\vec P_{AB}=\sum_i r_i \vec P_A^i\otimes \vec P_B^i$ with $\vec P_A^i\in \mathcal{S}_A$ and $\vec P_B^i\in \mathcal{S}_B$ normalized and pure and $r_i\in \mathds{R}$~\cite{Barrett07}. 

For a transformation on system $A$ defined by $\vec P_A\mapsto \vec P_{A'}=M_A\cdot \vec P_A$ the transformation of the joint system is given by $\vec P_{AB}\mapsto \vec P_{A'B}=(M_A\otimes \mathds{1})\cdot \vec P_{AB}$~\cite{Barrett07}. We demand that transformations $M_A$ on any system $A$ are {\em well-defined}, meaning $(M_A \otimes I_B) \cdot \vec{P}_{AB} \in \mathcal{S}_{AB}$ whenever $\vec{P}_{AB} \in \mathcal{S}_{AB}$ for all types of system $B$. 

In the following, we will always assume that the set of transformations allowed by the theory includes removing systems (which corresponds to taking the marginal state, as defined above) and adding a system, taking $\vec P_{A} \mapsto \vec P_{A}\otimes \vec P_{B}$.  

We also demand that the theory contains `classical' systems of type $(1,d)$ for all $d\in\mathds{N}$. We call the trivial classical system of type $(1,1)$ the vacuum ($V$). We shall in our proofs, taking inspiration from~\cite{BarnumBLW06}, use the fact that the state of a classical system can be {\em cloned}---see the technical supplement.

As shown e.g.\ in~\cite{Hardy01}, finite dimensional quantum theory as well as classical probability theory fit into this framework. So does {\em box-world}~\cite{Barrett07}. This allows all states on discrete sets of measurements that are non-signalling. The simplest non-trivial example of this is for elementary systems of type (2,2). The joint state space of two such systems includes PR-boxes. A key difference between box-world and quantum theory is that only the latter respects Tsirelson's bound.

{\em \bf Tsirelson's bound.---}The quantum correlation strength as quantified by the CHSH Bell inequality~\cite{ClauserHSH69} is upper bounded by Tsirelson's bound~\cite{Tsirelson93, Cirelson80}.
\begin{defin}[Tsirelson's bound]
Consider two systems $A$ and $B$, with two choices of measurements ($0$ or $1$) and two outputs each ($a$ and $b$). Define the quantity
 \begin{equation*}
S:= p(a=b|00)+p(a=b|01)+p(a=b|10)+p(a\neq b|11).
\end{equation*}
The theory governing the systems is said to satisfy Tsirelson's bound if $2-\sqrt{2} \leq S \leq 2+\sqrt{2}$ for any states allowed by the theory.
\end{defin}

A PR-box (also known as a non-local box) is designed to have S=0 or 4, thus maximally violating the Tsirelson bound~\cite{PopescuR94}. It is defined (up to relabellings of measurement choices and outcomes) to be a state where $$p(a=b|00)=p(a=b|01)=p(a=b|10)=p(a\neq b|11)=1$$ and the local marginal states are uniformly random.

{\em \bf Information causality.---}Let there be two space-like separated parties, Alice and Bob which share an arbitrary no-signaling resource. Alice then receives a random bit-string $\vec a=(a_0,\dots,a_{N-1})$, which is not known to Bob. The bits $a_i$ are unbiased and independently distributed. At the same time Bob gets a random variable $b\in\lbrace 0,\dots,N-1\rbrace$, which is unknown to Alice. Alice is free to make
use of her local resources in order to prepare a classical bit-string $\vec x$ of length $m$ which she sends to Bob. Bob, having received Alice's message, is then asked to guess the value of $a_b$ as best as he can. Let us denote Bob's guess by $\beta$. The efficiency of Alice's and Bob's strategy can be quantified by
$I\equiv \sum_{i}I_{\mathrm{Sh}}(a_i:\beta|b=i)$ where $I_{\mathrm{Sh}}(a_i:\beta|b=i)$ is the Shannon mutual information between $a_i$ and $\beta$, computed
under the condition that Bob has received $b=i$.
\begin{defin}[Information Causality]
A theory is said to respect information causality if in the above game $I\leq m$ for any allowed resource state.
\end{defin}

It was shown in~\cite{PawlowskiPKSWZ09} that information causality implies Tsirelson's bound.

{\em \bf General entropy definition.---}We now recount certain results from recent research into how to quantify entropy in general probabilistic theories~\cite{ShortW09, Barnum09, KimuraNI09}. We shall in particular use a definition of entropy for general theories from~\cite{ShortW09} which is  based on the Shannon entropy. This is highly analogous to how the von Neumann entropy generalises the Shannon entropy $H_{\mathrm{Sh}}(\vec P )=-\sum_i P_i\log P_i$ to the quantum case. The intuition is that the von Neumann entropy is the minimal Shannon entropy over all measurements. Actually it is over all {\em fine-grained} measurements (explained below).

Note that one can in general define the Shannon entropy associated with a measurement $e$ as $H_{\mathrm{Sh}}(e(\vec P))=-\sum_i (\vec R_i^e\cdot \vec P)\log (\vec R_i^e\cdot \vec P)$

\begin{defin}[Entropy \cite{ShortW09}]
For every normalized state $\vec P\in\mathcal S$ the entropy $H(\vec P)$ is given by
\begin{equation}
 \label{def_entr}
H(\vec P)\equiv \inf\limits_{e\in\mathcal{M^*}}H_{\mathrm{Sh}}\left(e(\vec P)\right)\, .
\end{equation}
$e(\vec P)$ denotes the classical probability distribution for the different outcomes of $e$ and 
the minimization is over the set of all fine-grained measurements $\mathcal{M^*}$. 
\end{defin}

$\mathcal{M^*}$ above is defined to be the set of measurements which have no {\em non-trivial fine-grainings}. A {\em fine-graining} is a subdivision of one outcome into several different outcomes. A {\em trivial} fine-graining is one where the resulting outcomes do not have independent probabilities, or more formally, where the vectors representing the respective effects are proportional to the effect-vector associated with the original coarse-grained outcome. 

The restriction to minimizing over $\mathcal{M^*}$ is important. If one allowed coarse-grained measurements the entropy could always be reduced arbitrarily by grouping outcomes together into single outcomes. It is natural to draw the line at trivial fine-grainings since no more information is yielded by them.

The entropy $H(\vec P)$ can be interpreted as the minimal uncertainty that
is associated with the outcome of a maximally informative measurement. 
It has some appealing properties: (i) $H$ reduces to the Shannon entropy for classical probability theory and the von
Neumann entropy in quantum theory, (ii) Suppose that the minimal number of outcomes for a fine-grained
measurement in $\mathcal{M}^*$ is $d$. Then for all states $\vec P \in \mathcal{S}$, $
 \log(d)\geq H(\vec P)\geq0$ and (iii) for any $\vec P_1$, $\vec P_2 \in \mathcal{S}$ and any mixed state 
$\vec P_{\rm mix}=p\vec P_1+(1-p)\vec P_2\in\mathcal{S}$:
$H(P_{\rm mix})\geq pH(\vec P_1)+(1-p)H(\vec P_2)$~\cite{ShortW09}.

For a state $\vec P_{AB}$ of a bipartite system $AB$ one defines the conditional entropy of $A$ conditioned on $B$ by~\cite{ShortW09}
\begin{equation}
 \label{condentr}
H(A|B)_{\vec P_{AB}}\equiv H(\vec P_{AB})-H(\vec P_{B})\,,
\end{equation}
with $\vec P_B$ the reduced state of $\vec P_{AB}$.
If there are no ambiguities we drop the indices and we write $H(A)$ instead of $H(\vec P_A)$ and
$H(AB)$ instead of $H(\vec P_{AB})$, and so on.

Some properties that are satisfied in quantum theory (where this
entropy reduces to the von Neumann entropy) are not necessarily satisfied for arbitrary theories. In box-world, for
example so-called strong subadditivity can be violated, as well as the subadditivity of the conditional entropy~\cite{ShortW09}.

{\em \bf Data processing inequality.---}The data processing inequality (DPI) is a crucial property of entropy measures which is frequently used in proofs in classical as well as quantum information theory~\cite{CoverT06, NielsenC00}. 
DPI quantifies the notion that local operations cannot increase correlations. 
A standard formulation for the classical case is that $H(X|Y)\leq H(X|g(Y))$, where $X$ and $Y$ are random variables which may be correlated, $H(X|Y):=H(XY)-H(Y)$, and $g(Y)$ is a function of Y only. The quantum DPI is the same, but with $H$ denoting the von Neumann entropy. 

We will here use the following generalised definition of DPI due to Short and Wehner~\cite{ShortW09}.
\begin{defin}[Data Processing Inequality (DPI)]
Consider two systems $A$ and $B$. The data processing inequality is that for any allowed state $\vec P_{AB} \in \mathcal{S}_{AB}$ and for any allowed
local transformation $T: \vec P_B \rightarrow \vec P'_B$ 
\begin{equation}
 \label{dataproc_entr}
H(A|B)_{\vec P_{AB}}\leq H(A|B')_{(\id\otimes T)\vec P_{AB}},
\end{equation}
 where $H(\cdot |\cdot)$ denotes the conditional entropy of Eqn.~\eqref{condentr}.
\end{defin}

{\em \bf Main result.---}Our main result links the data processing inequality with Tsirelson's bound. 

\begin{thm}
In any general probabilistic theory where the Data Processing Inequality is respected, the Tsirelson bound is respected. 
\end{thm}

\begin{proof}

We here sketch the proof---see the appendix for the details. 

We use the fact that the entropy of Def. 3 satisfies two properties: (i) $H(A|B):=H(AB)-H(B)$ (we call this COND), and (ii) it reduces to the Shannon entropy for classical systems (we call this SHAN).

We prove that for {\em any} theory and entropy measure $H$ jointly satisfying COND, SHAN and DPI, Tsirelson's bound holds (where DPI has been defined using $H$). This implies the main theorem. %(These conditions can be compared with those of~\cite{Barnum09}; note that there COND is taken as the definition of conditional entropy and thus also assumed.) 

The three conditions are not trivially applicable to restrict the resource state in van Dam's game so we use them, within the framework of probabilistic theories, to derive certain more directly applicable lemmas, including: (i) {\em $\sum_i H(A_i |\gamma)\geq H(A|\gamma)$, where $A_i$ denotes the i-th party of a multi-party system $A$} (ii) {\em  $H(A)\geq H(A|B)$ with equality for product states}, and (iii) {\em for classical systems $X$, $H(X|Y)\geq 0$}. With these lemmas and some additional arguments we show information causality is respected, and thus, by~\cite{PawlowskiPKSWZ09}, Tsirelson's bound.

\end{proof}

{\em \bf Discussion.---}We have shown that the generalised DPI implies Tsirelson's bound.  This addresses a question raised in~\cite{ShortW09}, namely in what manner enforcing generalised entropic relations restricts the set of possible theories. It also contributes to our understanding of why Bell-violations in quantum theory respect Tsirelson's bound.

As indicated in the proof sketch, our quantitative results can be applied to more general entropy measures. In particular, for {\em any} entropy measure $H$ and theory jointly satisfying COND, SHAN and DPI, we show that Tsirelson's bound holds. Thus one could alternatively have used for example the {\em decomposition entropy} of~\cite{ShortW09} in the statement of the main theorem as it satisfies SHAN and is defined to satisfy COND~\cite{ShortW09}. At the same time one may argue that whilst an operationally appealing definition of conditional entropy should automatically satisfy SHAN and DPI it is not clear why it should in general satisfy COND. COND may then be viewed as a {\em restriction} on states rather than a {\em definition} of conditional entropy.  

One can compare our three sufficient conditions COND, SHAN and DPI to those used in~\cite{PawlowskiPKSWZ09} and ~\cite{Barnum09} respectively. The entropic relations used in~\cite{PawlowskiPKSWZ09} to derive information causality were formulated in terms of a conditional mutual information $I(A:B|C)$. (It is assumed this can be defined in a more general setting, but no definition is given.) The conditions are that $I(A:B|C)$ should: be symmetric under change of A and B, be non-negative ($I\geq 0$), reduce to the Shannon mutual information for classical systems, obey the Data Processing Inequality as formulated for mutual information, and obey the chain rule $I(A:B|C)=I(A:BC)-I(A:C)$. Arguably our three relations are more minimalistic and natural than those. Moreover we show the arguments apply to particular concrete definitions of entropy and that for at least two particular definitions of conditional entropy DPI alone suffices. Consider secondly~\cite{Barnum09}. There  concrete entropy definitions are proposed and studied. The definitions are very similar to~\cite{ShortW09} though the framework is not a priori exactly identical. They define three properties in terms of conditional entropy as $H(AB)-H(B)$, with $H$ the measurement entropy: (i) 'monoentropicity' (two particular different entropy measures always have the same value), (ii) a version of the Holevo bound, and (iii) 'strong sub-additivity' (defined below). They show that those conditions imply information causality. They moreover note that conditions (ii) and (iii) can be derived from DPI defined in terms of the above conditional (measurement) entropy (more correctly they define it using mutual information $I(A:B):=H(A)+H(B)-H(AB)$ but this is equivalent in this case). Assumption (i) is used to obtain what we here derive as Eq.~\ref{eq:classical}. Thus it appears one may alternatively summarise their result on information causality as follows: DPI (in terms of COND and measurement entropy) plus mono-entropicity implies information causality. This can be compared to our Theorem 1; it is not so clear how to compare it to our more general Theorem 2, as the latter does not refer to a specific entropy measure, but to any state space and conditional entropy measure jointly satisfying DPI, COND and SHAN.

DPI is related to a condition known as {\em strong subadditivity} (SSA) which states that $H(A|CD)\leq H(A|C)$. SSA is {\em implied} by DPI since forgetting $D$ is an allowed local operation. In the quantum case SSA also implies DPI, but this does not necessarily hold in other theories as the standard quantum proof relies on the specific quantum feature known as Stinespring dilation. In the extreme case of box-world it was already known that SSA (and thus also DPI) is violated~\cite{ShortW09}. As an example consider two classical bits $x^0$, $x^1$ and a {\em gbit} $Z$. The latter is a (2,2) system which can take any allowed distributions, i.e.\ its state space is the convex hull of four states wherein the two outcomes take defined values for each measurement. The classical bits are uniformly random but the gbit contains their values. Then $H(x^0|x^1Z)=1$ whereas $H(x^0|Z)=0$, violating SSA~\cite{ShortW09}.

It is an open question whether there are theories which satisfy DPI but have states not contained in quantum theory, since Tsirelson's $2+\sqrt{2}$ bound is insufficient to rule out all non-quantum states. Understanding this and  with what DPI needs to be supplemented in order to derive quantum theory fully are natural next steps.

{\em \bf Acknowledgements.---}We acknowledge comments on an earlier draft from J. Oppenheim, A. Short and S. Wehner, advice on references by V. Scarani, as well as funding from the Swiss National Science Foundation (grant No. 200020-135048) and the European Research Council (grant No. 258932). The research was carried out in connection with DL's Master's thesis at ETH Zurich.

{\em Additional Note.---} Similar results have been obtained 
independently in~\cite{AlSafiS11} by Al-Safi and Short.

\bibliography{DPIrefs}
\bibliographystyle{h-physrev}

\newpage
\appendix

\section{Proof of main theorem}
The main theorem is a direct corollary of a more general theorem, Theorem 2, which we state and prove in this section. 
Crucially, Theorem 2 does not refer to a specific entropy measure such as the measurement entropy defined above. 

We require three definitions to state this theorem.

Firstly we redefine DPI, now defined without reference to a specific entropy definition.
\begin{defin}[Data Processing Inequality (DPI)]
Consider two systems $A$ and $B$. The data processing inequality is that for any allowed state $\vec P_{AB} \in \mathcal{S}_{AB}$ and for any allowed
local transformation $T: \vec P_B \rightarrow \vec P'_B$ 
\begin{equation}
 \label{dataproc_entr}
H(A|B)_{\vec P_{AB}}\leq H(A|B')_{(\id \otimes T) \vec P_{AB}}.
\end{equation}
\end{defin}

\begin{defin}[Conditional entropy (COND)]
The conditional entropy $H(A|B)$, however it is defined, must for all allowed states on $AB$ satisfy 
\begin{equation}
\label{eq:cond1}
H(A|B)=H(AB)-H(B).
\end{equation}
\end{defin}

\begin{defin}[Reduction to Shannon entropy (SHAN)]
The entropy $H$ must reduce to the Shannon entropy for classical systems.
\end{defin}

Our statements are restricted to the generalised probabilistic framework, as described in the introduction to the paper. We shall be making use of two non-trivial but operationally well-motivated types of transformations associated with that framework: {\em adding} and {\em removing} systems. An (independent) system in state $\vec{P_B}$ is {\em added} by the map taking any $\vec{P_A}$ to $\vec{P_A}\otimes \vec{P_B}$. A system is {\em removed} by taking the marginal distribution on the other system(s), as described in the introduction. We shall make use of the fact that this map acts to take the removed system $B$ to the vacuum system $V$. The only normalised state of the vacuum is $\vec{\id_V}=1$ (this can be seen from the equivalent definition of the marginal state used e.g. in~\cite{BarnumBLW06}). Thus, and this is another equation we shall find useful, $\vec{P_A}\otimes \vec{\id_V}=\vec{P_A}\,\forall \vec{P_A}$.

We shall also be assuming that the entropy measure is operational, i.e.\ is uniquely determined by the statistics of the experiment under consideration. Thus it is for a given set-up determined by the state of the systems under consideration. More subtly, $H$ moreover cannot depend on the order in which the state-spaces of the subsystems are composed, as this order is arbitrary; different observers describing the same experiment can make different choices here. Thus $H(AB)$ must be invariant under the interchange of systems $A$ and $B$.

We are now ready to state the theorem:

\begin{thm} For any probabilistic theory and entropy measure $H$ satisfying COND, SHAN and DPI, Tsirelson's bound holds.
\end{thm}

Before proving Theorem 2 we note that the main theorem (Theorem 1) is directly implied by this statement as the entropy $H$ referred to there satisfies COND and SHAN.

Before proving Theorem 2 we prove some lemmas which we shall need and which may be of interest in themselves.

\begin{lem}
\label{lem:AEC}
COND and DPI imply the relation 
\begin{equation}
\label{eq:cond3}
\sum_i H(A_i|\gamma)\geq H(A_1...A_n|\gamma)
 \end{equation}
for any $\vec{P}_{A_1...A_n} \in \mathcal{S}_{A_1...A_2} $, where $A_i$ denotes the i-th party of the total system $A_1...A_n$.
\end{lem}
\begin{proof}
Consider firstly $n=2$. By COND we have  $$H(A_1|\gamma)+H(A_2|\gamma)-H(A_1A_2|\gamma)=-H(A_2|A_1\gamma)+H(A_2|\gamma).$$
By DPI this is greater than or equal to 0. 

To generalise the argument to $n>2$, let $A_2$ be replaced by $A_2...A_n$ in  the previous equation. Then by the same argument 
$$H(A_1|\gamma)+H(A_2..A_n|\gamma)-H(A_1A_2..A_n|\gamma)\geq 0.$$
Now we can apply the previous argument to the term $H(A_2..A_n|\gamma)$ to get 
$$H(A_2|\gamma)+H(A_3...A_n|\gamma)\geq H(A_2...A_n|\gamma).$$ This process is then repeated iteratively to recover $\sum_i H(A_i|\gamma)\geq H(A_1...A_n|\gamma)$.
\end{proof}

\begin{lem}
\label{lem:ProdState}
For product states $\vec{P}_A\otimes \vec{P}_B$, COND, SHAN and DPI imply the relation
\begin{equation}
\label{eq:cond2}
 H(A|B)=H(A).
 \end{equation}
 \end{lem}

\begin{proof}
We firstly use COND and SHAN to show that $H(A|V)=H(A)$ for any system $A$. This follows from the following:
\begin{eqnarray}
H(A|V)&=& H(AV)-H(V)\\
      &=&H(A)-0\\
      &=&H(A).
\end{eqnarray}
(Here COND implies the first line. As $V$ is classical and with only one measurement outcome, SHAN implies $H(V)=0$; $\vec{P_{AV}}=\vec{P_A}$ as mentioned in the beginning of the appendix.) 

We now prove the equality of the lemma by separately proving the two corresponding inequalities in both directions. Note firstly that 
\begin{equation}
\label{eq:cond4}
H(A)\geq H(A|B)
\end{equation}
 for any state. To see this, consider the transformation $T$ that takes $B$ to the vacuum system (i.e. the transformation that {\em removes} $B$ as described in the introduction to the appendix). 
Then, using DPI, $$H(A|B)\leq H(A|T(B))=H(A|V)=H(A).$$

Consider secondly the inequality in the other direction, restricting ourselves to the case of product states only: 
\begin{equation}
\label{eq:prodineq}
H(A)_{\vec{P}_A} \leq H(A|B)_{\vec{P}_A\otimes \vec{P}_B}.
\end{equation}
This is true because $H(A)=H(A|V)_{\vec{P}_A\otimes \vec{I}_V}\leq H(A|B)_{\vec{P}_A\otimes \vec{P}_B}$, where the last step uses DPI for the transformation that creates $\vec{P}_B$ from the vacuum state (i.e. the transformation that {\em adds} $B$ as described in the introduction to the appendix).    

Combining Eqns. \ref{eq:cond4} and \ref{eq:prodineq} proves the claim.
\end{proof}

\begin{lem}
DPI, SHAN and COND imply that for all classical systems X,  
\begin{equation}
\label{eq:classical}
H(X|Y)\geq 0.
\end{equation}

\end{lem}

\begin{proof}
To prove the lemma via DPI we shall use the fact that the extremal states of classical systems can be {\em cloned}~\cite{BarnumBLW06}. More specifically,  we shall make use of the fact that for a classical system $X_A$ in state $\vec{P}_A=\sum_ip_i\vec{\mu}_i$, where the $\vec{\mu}$ are pure, and another classical system $X_B$ of the same dimensionality in any given independent pure state $\vec{\mu}_k$, there exists a map $T_C$ such that $T_C(\vec{P}_A\otimes\vec{P}_B)=\sum_ip_i\vec{\mu}_i\otimes\vec{\mu}_i$.

We shall consider a three-party system $YX_AX_B$, where $Y$ is the only non-classical sub-system. The idea is that given an arbitrary state on $YX := YX_A$, we can always bring in another independent subsystem $X_B$ and perform a cloning operation so that $X_B$ becomes a copy of $X_A$. We may then apply DPI on the cloning transformation $T_C$ applied on $X_A$ and $X_B$. 
We call the states before and after the cloning $\vec{P}_{YX_AX_B}^i$ and $\vec{P}_{YX_AX_B}^f$ respectively.

By DPI we then have
\begin{equation}
H(Y|X_AX_B)_{\vec{P}_{YX_AX_B}^i}\leq   H(Y|X_AX_B)_{\vec{P}_{YX_AX_B}^f}
\end{equation}
Note now that the left-hand-side can be simplified. COND together with Eq.~(\ref{eq:cond2}) imply that $H(AB)=H(A)+H(B)$ for independently prepared $A$ and $B$. This can be applied here because $X_B$ is initially in an independent state, yielding  
$$H(Y|X_AX_B)_{\vec{P}_{YX_AX_B}^i}=H(Y|X_A)_{\vec{P}_{YX_AX_B}^i}.$$
Accordingly  
$$H(Y|X_A)_{\vec{P}_{YX_AX_B}^i}\leq H(Y|X_AX_B)_{\vec{P}_{YX_AX_B}^f}.$$

We also note that the marginal state on $YX_A$ is unchanged by the cloning, i.e. $\vec{P}_{YX_A}^i=\vec{P}_{YX_A}^f$, so we may for simplicity write that for the state {\em after} the cloning, 
\begin{equation}
\label{eq:wq}
H(Y|X_A)\leq H(Y|X_AX_B).
\end{equation}
In the following, unless stated otherwise, we consider the state after the cloning only.

Applying Eq.~(\ref{eq:cond1}), i.e. COND, to Eq.~(\ref{eq:wq}) and undertaking some rearrangements yields
$$H(X_B|YX_A)\geq H(X_B|X_A).$$
Moreover, SHAN implies that $H(X_B|X_A)=0$. Thus  
$$H(X_B|YX_A)\geq 0.$$

Note that since $X_A$ and $X_B$ are operationally indistinguishable after the cloning, $H(X_A|YX_B)=H(X_B|YX_A)$. Thus we have  
\begin{equation}
H(X_A|YX_B)\geq 0.
\end{equation}

By DPI 
\begin{equation}
H(X_A|Y)\geq H(X_A|YX_B). 
\end{equation}
Thus, still {\em after} the cloning, we have that  
\begin{equation}
H(X_A|Y)\geq 0.
\end{equation}
But since the state of $X_AY$ is unchanged by the cloning transformation, this implies that the equation holds also for the (arbitrary) initial state of $X_AY$. Recall that we used $X_A$ to label the classical system $X$. We have thus shown that  $H(X|Y)\geq 0$ for an arbitrary initial state on $XY$.
\end{proof}

\begin{lem}
\label{lem:Renatos}
COND, SHAN and DPI imply the relation
\begin{equation}
\label{eq:Renatos}
 H(\vec a|B \vec x)\geq n-m,
 \end{equation}
 where the quantities are as defined in the information causality game ($\vec a$ is the classical $n$-bit string given to Alice, $B$ is the non-classical resource and $\vec x$ is the classical $m$-bit message sent to Bob). 
\end{lem}
\begin{proof}
\begin{eqnarray*}
H(\vec a|B \vec x)-H(\vec x|\vec aB)&=&-H(B \vec x)+H(\vec aB)\\
&=& -H(B \vec x)+H(\vec a)+H(B)\\
&=& H(\vec a)-H(\vec x|B)\\
&\geq & H(\vec a)-H(\vec x)\\
&=& n-H(\vec x)\\
&\geq & n-m.
\end{eqnarray*}
The first line follows from COND. The second line is due to the combination of Eq.~(\ref{eq:cond2}) and Eq.~(\ref{eq:cond1}) and recalling that $\vec a$ and $B$ are independent. The third line uses Eq.~(\ref{eq:cond1}) again. The fourth line follows from Eq.~(\ref{eq:cond4}). The fifth and sixth lines follow from the definition of the game as well as elementary properties of the Shannon entropy, which can be exploited due to SHAN.

It follows by applying Eq.~(\ref{eq:classical}) to the left hand side that $H(\vec a|B \vec x)\geq n-m$.
\end{proof}

We now put together the pieces to prove Theorem 2:
\begin{proof}[Proof of Theorem 2] By lemma \ref{lem:Renatos} above, we have $$ H(\vec a|B \vec x)\geq n-m. $$ By lemma \ref{lem:AEC} this implies 
$$\sum_i H(a_i|B \vec x)\geq n-m.$$
By DPI we accordingly have that for Bob's guess $\beta=\beta(B, \vec x, i)$
$$\sum_i H(a_i|\beta(i))\geq n-m,$$
where, by SHAN and the fact that $a_i$ and $\beta(i)$ are both classical, $H$ refers to the Shannon entropy. 

This implies information causality, as 
$I_{\mathrm{Sh}}(a_i:\beta(i))=H(a_i)-H(a_i|\beta(i))$, so
\begin{eqnarray*} 
\sum_i I_{\mathrm{Sh}}(a_i:\beta(i))&=&\sum_i H(a_i)-H(a_i|\beta(i))\\
&=&n-\sum_i H(a_i|\beta(i))\\
&\leq & m.
\end{eqnarray*}
Recall that information causality implies Tsirelson's bound.

\end{proof}

\end{document}